\documentclass[11pt]{amsart}

\usepackage{amssymb,latexsym,amsmath,extarrows, mathrsfs, amsthm,color,amsrefs}
\usepackage{graphicx}
\usepackage{MnSymbol}

\usepackage[margin=1in, centering]{geometry}
\usepackage[bookmarksnumbered, colorlinks, plainpages]{hyperref}
\usepackage{hyperref} %,hypertexnames=false,colorlinks,[pagebackref]
\hypersetup{
    colorlinks=true,       % false: boxed links; true: colored links
    linkcolor=blue,          % color of internal links
    citecolor=magenta,        % color of links to bibliography
    filecolor=magenta,      % color of file links
    urlcolor=cyan           % color of external links
}

\usepackage{mathtools}
\mathtoolsset{showonlyrefs,showmanualtags}

\newtheorem{theorem}{Theorem}[section]

\newtheorem{lemma}[theorem]{Lemma}

\newtheorem{remark}[theorem]{Remark}

\newtheorem{corollary}[theorem]{Corollary}

\newcommand{\Z}{\mathbb Z}
\newcommand{\R}{\mathbb R}
\newcommand{\C}{\mathbb C}
\newcommand{\T}{\mathbb T}
\newcommand{\N}{\mathbb N}
\newcommand{\Q}{\mathbb Q}

\definecolor{deepgreen}{cmyk}{1,0,1,0.5}
\newcommand{\Blue}[1]{{\color{blue} #1}}

\title{H\"older continuity of the integrated density of states for Liouville frequencies}

\author[R.\ Han]{Rui Han}
\address{Department of Mathematics \\ Louisiana State University  \\  Baton Rouge, LA 70803, USA}
\email{rhan@lsu.edu}

\author[W.\ Schlag]{Wilhelm Schlag}
\address{Department of Mathematics \\ Yale University \\ New Haven, CT 06511, USA}
\email{wilhelm.schlag@yale.edu}

\thanks{
R.\ Han is partially supported by NSF DMS-2143369. 
W.\ Schlag is partially supported by NSF grant DMS-2054841.
}
\begin{document}

\begin{abstract}
We prove H\"older continuity of the Lyapunov exponent $L(\omega,E)$ and the integrated density of states at energies that satisfy $L(\omega,E)>4\kappa(\omega,E)\cdot \beta(\omega)\geq 0$ for general analytic potentials, with $\kappa(\omega,E)$ being Avila's acceleration.
\end{abstract}

\maketitle 

\centerline{ In memory of Michael Goldstein}

\section{Introduction} 
We study one-dimensional quasi-periodic Schr\"odinger operators with analytic potentials:
\begin{align}\label{def:operator}
    (H_{\omega,\theta}\phi)_n=\phi_{n+1}+\phi_{n-1}+v(\theta+n\omega) \phi_n,
\end{align}
in which $\theta\in \T$ is called the phase, $\omega\in \T\setminus\Q$ is called the frequency.
Throughout the paper, for $\theta\in \T$, $\|\theta\|_{\T}:=\mathrm{dist}(\theta,\Z)$, and we shall simply write it as $\|\theta\|$.

Our result concerns the regularity of the integrated density of states in energy for quasi-periodic Schr\"odinger operators.
For integers $a<b$ and interval $\Lambda=[a,b]\cap \Z$, let $H_{\omega,\theta}|_{\Lambda}$ be the restriction of $H_{\omega,\theta}$ to the interval $\Lambda$ with Dirichlet boundary condition.
Let $\{E_{\Lambda,j}(\omega,\theta)_{j=1}^{|\Lambda|}$ be the eigenvalues of $H_{\omega,\theta}|_{\Lambda}$.
Consider
\begin{align}
    N_{\Lambda}(\omega,E,\theta):=\frac{1}{|\Lambda|}\sum_{j=1}^{|\Lambda|}\chi_{(-\infty,E)}(E_{\Lambda,j}(\omega,\theta)).
\end{align}
It is well-known that the weak-limit
\begin{align}
    \lim_{a\to-\infty, b\to \infty}\mathrm{d} N_{\Lambda}(\omega,E,\theta)=:\mathrm{d} \mathcal{N}(\omega,E).
\end{align}
exists and is independent of $\theta$.
The distribution $\mathcal{N}(\omega,\cdot)$ is called the integrated density of states (IDS). It is connected to the Lyapunov exponent $L(\omega,E)$, see \eqref{def:LE_operator}, through the Thouless formula:
\begin{align}\label{eq:Thouless}
    L(\omega,E)=\int \log |E-E'|\, \mathrm{d}\mathcal{N}(\omega,E').
\end{align}

Concerning the regularity of $\mathcal{N}(\omega,E)$ in $E$, Craig and Simon \cite{CS} established log-H\"older continuity for general ergodic operators.
Goldstein and the second author \cite{GS1} established the first H\"older continuity result for general analytic potentials on $\T$ throughout the positive Lyapunov exponent regime for Diophantine frequencies, as well as a weaker regularity for general $\T^d$, $d\geq 2$.
The Diophantine condition was relaxed to cover some weak Liouville $\omega$'s in \cite{YZ} and then further strengthened in \cite{HZ} to cover all the $\omega$'s such that $0\leq \beta(\omega)<c_v L(\omega,E)$, where $\beta(\omega)$ is as in \eqref{def:beta} and $c_v>0$ is an implicit constant.
The results of \cite{YZ,HZ} are both for the positive Lyapunov exponent regime, and follow the scheme of Goldstein and the second author~\cite{GS1}: combining two important tools, the large deviation estimates (LDT) and the Avalanche Principle (AP).
The former was first introduced in the seminal paper of Bourgain and Goldstein \cite{BG} with significant further developments by Bourgain, Goldstein and the second author in \cite{BGS1,BGS2,Bo,GS2,GS3}, and the latter was introduced in \cite{GS1} and has been being greatly expanded by Duarte and Klein \cite{DK1,DK2}.

In the regime of small coupling constants (which is a sub-regime of the zero Lyapunov exponent regime), $1/2$ H\"older continuity was established using almost reducibility by Avila-Jitomirskaya \cite{AJ} for Diophantine frequencies. 
For a special family, the almost Mathieu operator with $v(\theta)=2\lambda\cos(2\pi\theta)$ and Diophantine frequencies, in the same work, $1/2$-H\"older was proved for any non-critical coupling constant $\lambda\neq 1$.

In a recent paper, Avila-Last-Shamis-Zhou \cite{ALSZ} proved the regularity of IDS for the almost Mathieu operator can be as poor as only log-H\"older in the regime $\beta(\omega)>1.5 \log \lambda\geq 0$, where $\log\lambda=L(\omega,E)$ when $\lambda\geq 1$ on the spectrum.

The main result of this paper is the following on an opposite regime to \cite{ALSZ}, applying to general analytic potentials in the positive Lyapunov exponent regime.
\begin{theorem}\label{thm:main}
    Let $v$ be an analytic function on $\T$. Let $\beta(\omega)$ be as in \eqref{def:beta}, $L(\omega,E)$ is the Lyapunov exponent and $\kappa(\omega,E)$ is Avila's acceleration as in \eqref{def:acceleration}. 
    Then for $E$ such that $$L(\omega,E)>4\kappa(\omega,E)\cdot \beta(\omega),$$ the integrated density of states $\mathcal{N}(\omega,E)$ is H\"older continuous at $E$.
\end{theorem}
\begin{remark}
By the Thouless formula, the same H\"older continuity holds for $L(\omega,E)$.
\end{remark}
Theorem \ref{thm:main} strengthens the result of \cite{HZ} by replacing the implicit condition $L(\omega,E)>C_v\beta(\omega)$ therein with an effective condition.
The main reason for such improvement is the recent established connection between the LDT theory and Avila's acceleration in \cite{HS1} and its developments in various settings \cite{HS2,HS3,Ha}.
In addition to these ingredients, the technical core of this paper is Theorem \ref{thm:LDT}, an effective version of LDT for Liouville frequencies.
The mechanism for bounding the IDS proceeds as in \cite{GS2} by Goldstein and the second author: one controls the number of eigenvalues falling into an interval $I$ by estimating the trace of the finite volume resolvent  for energies $E=E_0+i\epsilon$, where $\epsilon=|I|$.

We introduce some notations before we proceed.
For any $R>1$, let $\mathcal{A}_R:=\{z\in \C: 1/R<|z|<R\}$ be the annulus. Let $\mathcal{C}_r:=\{z\in \C: |z|=r\}$ be the circle with radius $r>0$.
For $z\in \C$ and $r>0$, let $B_r(z)$ be the open ball centered at $z$ with radius $r$.
For $x\in \R$, let $[x]$ be the largest integer such that $[x]\leq x$.
Throughout this paper, for an analytic function $g(z)$ on $\mathcal{A}_R$, let 
$$\hat{g}(k):=\int_{\T} g(e^{2\pi i\theta}) e^{-2\pi ik\theta}\, \mathrm{d}\theta, \text{ for } k\in \Z,$$
be the Fourier coefficients.

We organize the rest of the paper as follows: Sec. 2 serves as preliminaries. We present a quick proof of Theorem \ref{thm:main} in Sec. 3, assuming the large deviation Theorem \ref{thm:LDT}, whose proof is postponed to Sec. 4 and 5.

\section{Preliminaries}\label{sec:pre}

\subsection{Continued fraction expansion}\label{sec:continued_fraction}
Give $\omega\in (0,1)$, let $[a_1,a_2,...]$ be the continued fraction expansion of $\omega$. For $n\geq 1$, let $p_n/q_n=[a_1,a_2,...,a_n]$ be the continued fraction approximants. 
The following property is well-known for $n\geq 1$,
\begin{align}\label{eq:qn_omega_min}
\|q_n\omega\|_{\T}=\min_{1\leq k<q_{n+1}} \|k\omega\|_{\T},
\end{align}
The $\beta(\omega)$ exponent measures the exponential closeness of $\omega$ to rational numbers:
\begin{align}\label{def:beta}
\beta(\omega):=\limsup_{n\to\infty}\frac{\log q_{n+1}}{q_n}=\limsup_{n\to\infty}\left(-\frac{\log \|n\omega\|_{\T}}{n}\right).
\end{align}
Let 
\begin{align}\label{def:betan}
\beta_n(\omega):=\frac{\log q_{n+1}}{q_n}. 
\end{align}
It is well-known that
\begin{align}\label{eq:qn_omega}
\|q_n\omega\|\in (1/(2q_{n+1}), 1/q_{n+1})=(e^{-\beta_n(\omega) q_n}/2, e^{-\beta_n(\omega) q_n}).
\end{align}
The following are standard:
\begin{align}\label{eq:qn+1=qn+qn-1}
    q_{n+1}=a_{n+1}q_n+q_{n-1}, \text{ and } p_{n+1}=a_{n+1}p_n+p_{n-1},
\end{align}
and
\begin{align}\label{eq:qn-1_norm=qn+qn+1}
\|q_{n-1}\omega\|=a_{n+1}\|q_n\omega\|+\|q_{n+1}\omega\|.
\end{align}
To see the latter is true, note that by \eqref{eq:qn+1=qn+qn-1},
\begin{align}
    q_{n+1}\omega-p_{n+1}=a_{n+1}(q_n\omega-p_n)+(q_{n-1}\omega-p_{n-1}),
\end{align}
which yields \eqref{eq:qn-1_norm=qn+qn+1} since $\|q_n\omega\|=|q_n\omega-p_n|$ and
$$(q_{n+1}\omega-p_{n+1})\cdot (q_n\omega-p_n)<0,$$
for any $n\geq 1$.

\subsection{Green's function for the annulus.}We state the precise Green's kernel, which can be derived  by the method of images. 

\begin{lemma}\label{lem:Green_AR}\cite[Lemma 3.1]{HS1}
The Green's function on the annulus $\mathcal{A}_R$ is given by
\begin{align}\label{def:G}
G_R(z,w)=\frac{1}{2\pi}\log |z-w| + \Gamma_R(z,w),
\end{align}
where
\begin{align}\label{eq:tildeG_prod}
\Gamma_R(z,w)=\frac{\log( |z|/ R) \log (|w|/ R)}{4\pi\log R}+\frac{1}{2\pi}\log \left( \frac{\prod_{k=1}^{\infty} |1-\frac{1}{R^{4k}}\frac{z}{w}| \cdot |1-\frac{1}{R^{4k}} \frac{w}{z}|}{R\cdot \prod_{k=1}^{\infty} |1-\frac{1}{R^{4k-2}}w\overline{z}|\cdot |1-\frac{1}{R^{4k-2}} \frac{1}{\overline{z}w}|}\right).
\end{align}
The Green's function is symmetric and invariant under rotations: $G_R(z,w)=G_R(w,z)$ and $G_R(z,w)=G_R(e^{i\phi} z,e^{i\phi} w)$. 
\end{lemma}
It is also easy to check that
\begin{align}\label{eq:GR_even}
    &2\pi G_R(1/\overline{z}, 1/\overline{w})\\
    =&
    \log |1/\overline{z}-1/\overline{w}|+\frac{\log (1/(R|z|))\log (1/(R|w|))}{2\log R}
    +\log \left( \frac{\prod_{k=1}^{\infty} |1-\frac{1}{R^{4k}}\frac{\overline{w}}{\overline{z}}| \cdot |1-\frac{1}{R^{4k}} \frac{\overline{z}}{\overline{w}}|}{R\cdot \prod_{k=1}^{\infty} |1-\frac{1}{R^{4k-2}}\frac{1}{\overline{w}z}|\cdot |1-\frac{1}{R^{4k-2}} z\overline{w}|}\right)\\
    =&2\pi G_R(z,w).
\end{align}

The following integral is useful, see \cite[Lemma 3.2]{HS1}, note $\Gamma_R$ here is $H_R$ therein.
\begin{align}\label{eq:int_HR}
\int_0^{1}   \Gamma_R(re^{ 2\pi i\theta},w) \, \mathrm{d}\theta= \frac{\log(r/R)}{4\pi\log  R} \log(|w|/R) - \frac{\log R}{2\pi}.
\end{align}

\subsection{Cocycles and Lyapunov exponents}

Let $(\omega, A)\in (\T, C^{\omega}(\T, \mathrm{SL}(2,\R)))$. 
Let 
\begin{align}
A_n(\omega,\theta)=A(\theta+(n-1)\omega)\cdots A(\theta).
\end{align}
Let the finite-scale Lyapunov exponents be defined as 
\begin{align}\label{def:LE}
L_n(\omega, A):=\frac{1}{n}\int_{\T} \log \|A_n(\omega,\theta)\|\, \mathrm{d}\theta,
\end{align} 
and the infinite-scale Lyapunov exponents as
\begin{align}
L(\omega, A)=\lim_{n\to\infty}L_n(\omega, A).
\end{align}

We denote the phase-complexified Lyapunov exponents as $$L_n(\omega,A(\cdot+i\varepsilon))=:L_n(\omega,A,\varepsilon), \text{ and }L(\omega,A(\cdot+i\varepsilon))=:L(\omega,A,\varepsilon),$$ respectively.

One can rewrite the eigenvalue equation $H_{\omega,\theta}\phi=E\phi$ of \eqref{def:operator} as follows:
\begin{align}
    \left(\begin{matrix} \phi_{n+1}\\ \phi_n\end{matrix}\right)=M_E(\theta+n\omega) 
    \left(\begin{matrix} \phi_n \\ \phi_{n-1}\end{matrix}\right),
\end{align}
in which
\begin{align}
    M_E(\theta)=\left(\begin{matrix}E-v(\theta+n\omega)\ &-1\\ 1 &0\end{matrix}\right),
\end{align}
is the one-step transfer matrix. Let $M_{m,E}(\theta):=M_E(\theta+(m-1)\omega)\cdots M_E(\theta)$ be the $m$-step transfer matrix. 
The following relation between $M_{m,E}$ and the Dirichlet determinants is well-known:
\begin{align}\label{eq:M=PPPP}
    M_{m,E}(\theta)=\left(\begin{matrix} P_{m,E}(\theta) &-P_{m-1,E}(\theta+\omega)\\
    P_{m-1,E}(\theta) &-P_{m-2,E}(\theta+\omega)\end{matrix}\right),
\end{align}
in which $P_{m,E}(\theta)=\det(H_{\omega,\theta}|_{[0,m-1]}-E)$ and $H_{\omega,\theta}|_{[0,m-1]}$ is the restriction of operator $H_{\omega,\theta}$ to the interval $[0,m-1]$ with Dirichlet boundary condition.

For the given transfer matrix $M_E$, we denote
\begin{align}\label{def:LE_operator}
L(\omega,M_E)=:L(\omega,E),
\end{align}
for simplicity.

\subsection{Avila's acceleration}
Let $(\omega, A)\in (\T, C^{\omega}(\T, \mathrm{SL}(2, \R)))$.
The Lyapunov exponent $L(\omega,A,\varepsilon)$ is a convex and even function in $\varepsilon$.
Avila defined the acceleration to be the right-derivative as follows:
\begin{align}\label{def:acceleration}
\kappa(\omega, A,\varepsilon):=\lim_{\varepsilon'\to 0^+} \frac{L(\omega, A,\varepsilon+\varepsilon')-L(\omega, A,\varepsilon)}{2\pi \varepsilon'}.
\end{align}
As a cornerstone of his global theory \cite{Global}, he showed that for analytic $A\in \mathrm{SL}(2,\R)$ and irrational $\omega$, $\kappa(\omega, A,\varepsilon)\in \Z$ is always quantized.

Recall that $v$ is an analytic functions on $\T_{\varepsilon_0}$ for some $\varepsilon_0>0$.  We may shrink $\varepsilon_0$ when necessary such that 
\begin{align}\label{eq:L_linear}
L(\omega,E,\varepsilon)=L(\omega,E,0)+2\pi \kappa(\omega,E,0)\cdot |\varepsilon|
\end{align}
holds for any $|\varepsilon|\leq \varepsilon_0$. 
For the rest of the paper, when $\varepsilon=0$, we shall omit $\varepsilon$ from various notations of Lyapunov exponents and accelerations.

\subsection{Upper semi-continuity}
The following upper bound is standard.
\begin{lemma}\label{lem:upper_semi_cont}
For any irrational $\omega$ and continuous cocycle $(\omega,A)\in (\T, C(\T,\mathrm{SL}(2,\R)))$, for any small $\delta>0$, and $n$ large enough, there holds:
\begin{align}
\frac{1}{n} \log \left\|A_n(\omega,\theta)\right\|\leq L(\omega,A)+\delta,
\end{align}
uniformly in $\theta\in \T$.
\end{lemma}

\subsection{Basic estimates in large deviation estimates}\label{sec:LDT}
Let 
\begin{align}
    u_{m,E}(\theta):=\frac{1}{m}\log \|M_{m,E}(\theta)\|\, \mathrm{d}\theta.
\end{align}
It is easy to see the following holds for some constant $C_{v,1}=C(\|v\|_{\T_{\varepsilon_0}})>0$:
\begin{align}\label{eq:vm_shift_invariance}
\left|u_{m,E}(\theta)-u_{m,E}(\theta+\omega)\right|\leq \frac{C_{v,1}}{m},
\end{align}

The Fej\'er kernel is
\[
F_Q(k)=\sum_{|j|<Q}\frac{Q-|j|}{Q^2}e^{2\pi i kj\omega}.
\]
The following estimates can be found in \cite{HZ}, with their proofs in \cite[Appendix E]{HZ}.
Below, $p/q$ is an arbitrary continued fractional approximant of $\omega$, as in \eqref{sec:continued_fraction}.
\begin{align}\label{eq:FR_1}
0\leq F_Q(k)\leq \min(1, \frac{2}{1+Q^2\|k\omega\|^2}),
\end{align}
\begin{align}\label{eq:FR_2}
    \sum_{1\leq |k|<q/4} \frac{1}{1+Q^2\|k\omega\|^2}\leq 2\pi \frac{q}{Q},
\end{align}
and
\begin{align}\label{eq:FR_3}
    \sum_{\ell q/4\leq k<(\ell+1)q/4} \frac{1}{1+Q^2\|k\omega\|^2}\leq 2+2\pi \frac{q}{Q}.
\end{align}

We will always use \eqref{eq:FR_1} to bound the Fej\'er kernel without explicitly referring to it in Sec. 5.

We also have two basic estimates for the Fourier coefficients:
\begin{align}\label{eq:Fourier_1/k}
|\hat{u}_{m,E}(k)|\leq \frac{C_{v,2}}{|k|},\,  k\neq 0,
\end{align}
for some constant $C_{v,2}=C(\|v\|_{\T_{\varepsilon_0}},\varepsilon_0)>0$. 

The next lemma was first proved in \cite{HZ}, and is useful in particular for $k$ not large.
\begin{lemma}\cite[Lemma 2.4]{HZ}\label{lem:hatu_small_k}
\begin{align}\label{eq:Fourier_1/mk}
|\hat{u}_{m,E}(k)|\leq \frac{C_{v,1}}{4m \|k\omega\|},\, k \neq 0.
\end{align}
\end{lemma}

\section{H\"older continuity of IDS}
Our proof of H\"older continuity is built on a large deviation estimate. The setup is as follows.
Let $E\in \R$ be an energy such that $L(\omega,E)>4\kappa(\omega,E)\cdot \beta(\omega)$. 
Let $\varepsilon_0>0$ be such that 
\begin{align}\label{eq:L_eps=L_0+kappa}
    L(\omega,E,\varepsilon)=L(\omega,E)+2\pi \kappa(\omega,E)|\varepsilon|, \text{ for } |\varepsilon|\leq \varepsilon_0.
\end{align}
Let $\eta_0>0$ be a small constant such that for any $|\eta|\leq \eta_0$ (recall that the acceleration $\kappa$ is upper semi-continuous with respect to the cocycle \cite{Global})
\begin{align}\label{def:eta_0}
    \begin{cases}\kappa(\omega,E+i\eta,\varepsilon)\leq \kappa(\omega,E,\varepsilon)=\kappa(\omega,E), \text{ for } |\varepsilon|\leq \varepsilon_0, \text{ and } \\
    L(\omega,E+i\eta)>4\kappa(\omega,E)\cdot \beta(\omega).
    \end{cases}
\end{align}
We now state the large deviation estimates:
\begin{theorem}\label{thm:LDT}
    For $E\in \R$, for any small $\delta\in (0, 1/(\beta(\omega)+1))$, for $m$ large enough, 
    \begin{align}
        \mathrm{mes}(\mathcal{B}_m):=\mathrm{mes}\{\theta: |u_{m,E}(\theta)-L_m(\omega,E)|>2\kappa(\omega,E)\beta(\omega)+C_{\omega,v}\delta\}\leq e^{-\delta^4 m},
    \end{align}
    for some constant $C_{\omega,v}>1$.
\end{theorem}

As a corollary of Theorem \ref{thm:LDT}, we have
\begin{corollary}\label{cor:LDT}
Let $E,\delta$ be as in Theorem \ref{thm:LDT}, for large enough $m$, and small $\eta$ such that
\begin{align}\label{eq:eta_small}
|\eta|\leq \exp(-2m(\kappa(\omega,E)\beta(\omega)+2C_{\omega,v}\delta)),
\end{align}
we have
\begin{align}
    \mathrm{mes}(\mathcal{B}_m^{\eta}):=\mathrm{mes}\{\theta: u_{m,E+i\eta}(\theta)<L(\omega,E)-2\kappa(\omega,E)\beta(\omega)-2C_{\omega,v}\delta\}\leq e^{-\delta^4 m}.
\end{align}    
\end{corollary}
We give a quick proof of this corollary. Within the proof we write $\kappa(\omega,E)=\kappa$ and $\beta(\omega)=\beta$.
\begin{proof}
    It suffices to prove $(\mathcal{B}_m)^c\subseteq (\mathcal{B}_m^{\eta})^c$.
    Taking any $\theta\in (\mathcal{B}_m)^c$, we have
    \begin{align}\label{eq:Mm_large}
        \|M_{m,E}(\theta)\|
        \geq &\exp(m(L_m(\omega,E)-2\kappa\beta-C_{\omega,v}\delta))\\
        \geq &\exp(m(L(\omega,E)-2\kappa\beta-C_{\omega,v}\delta)).
    \end{align}
    By telescoping and upper-semi-continuity (Lemma \ref{lem:upper_semi_cont}), we have for $\eta$ satisfying \eqref{eq:eta_small} that 
    \begin{align}\label{eq:Mm-Mm^eta}
        \left|\|M_{m,E}(\theta)\|-\|M_{m,E+i\eta}(\theta)\|\right|
        \leq &\|M_{m,E}(\theta)-M_{m,E+i\eta}(\theta)\|\\
        \leq &\exp(m(L(\omega,E)+\delta))\cdot |\eta|\\
        \leq &\exp(m(L(\omega,E)-2\kappa\beta-3C_{\omega,v}\delta)).
    \end{align}
    Combining \eqref{eq:Mm_large} with \eqref{eq:Mm-Mm^eta}, we have
    \begin{align}
        u_{m,E+i\eta}(\theta)
        =&u_{m,E}(\theta)+\frac{1}{m}\log \frac{\|M_{m,E+i\eta}(\theta)\|}{\|M_{m,E}(\theta)\|}\\
        =&u_{m,E}(\theta)+\frac{1}{m}\log \left(1+\frac{\|M_{m,E+i\eta}(\theta)\|-\|M_{m,E}(\theta)\|}{\|M_{m,E}(\theta)\|}\right)\\
        \geq &u_{m,E}(\theta)-\frac{2}{m}\frac{\left|\|M_{m,E}(\theta)\|-\|M_{m,E+i\eta}(\theta)\|\right|}{\|M_{m,E}(\theta)\|}\\
        \geq &u_{m,E}-\delta,
    \end{align}
    which yields $\theta\in (\mathcal{B}_m^{\eta})^c$.
\end{proof}

We postpone the proof of Theorem \ref{thm:LDT} to Sec. 5. Below we show how to utilize it to obtain Theorem \ref{thm:main}.
\subsection*{Proof of Theorem \ref{thm:main}}
Our goal is to analyze the following for an arbitrary fixed $\theta$ as $N\to\infty$:
\begin{align}\label{eq:goal}
    d_{\omega,N}(\theta,E-\eta,E+\eta):=&N_{[0,N-1]}(\omega,E+\eta,\theta)-N_{[0,N-1]}(\omega,E-\eta,\theta)\\
    =&\frac{1}{N}\cdot \#\{\sigma(H_{\omega,\theta}|_{[0,N-1]})\cap [E-\eta,E+\eta)\}\\
    =&\frac{1}{N}\mathrm{tr}(P_{[E-\eta, E+\eta)}(H_{\omega,\theta}|_{[0,N-1]})),
\end{align}
in which $H_{\omega,\theta}|_{[0,N-1]}$ is $H_{\omega,\theta}$ restricted to the interval $[0,N-1]$ with Dirichlet boundary condition.
Throughout the rest of this section, we shall omit $\omega$ in various notations.
Let $\delta>0$ be small such that
\begin{align}\label{def:delta_small}
    0<\delta<\min\left(\frac{1}{\beta(\omega)+1}, \frac{\min_{|\eta|\leq \eta_0} L(\omega,E+i\eta)-4\kappa(\omega,E)\beta(\omega)}{10C_{\omega,v}}\right),
\end{align}
in which $C_{\omega,v}$ is the constant in Theorem \ref{thm:LDT}.

We proceed as in \cite{GS2}, bounding \eqref{eq:goal} through trace of the Green's function using
\begin{align}
    \chi_{[E-\eta,E+\eta)}(x)\leq \frac{2\eta^2}{\eta^2+(x-E)^2}=2\eta \cdot \mathrm{Im}\left(\frac{1}{x-(E+i\eta)}\right).
\end{align}
This implies
\begin{align}\label{eq:dN<G}
    d_N(\theta,E-\eta,E+\eta)
    \leq &\frac{2\eta}{N} \mathrm{Im}(\mathrm{tr}(G_{[0,N-1]}^{E+i\eta}(\theta)))\\
    =&\frac{2\eta}{N}\sum_{k=0}^{N-1} \mathrm{Im} (G_{[0,N-1]}^{E+i\eta}(\theta;k,k)),
\end{align}
in which $G^{E+i\eta}_{[0,N-1]}(\theta)=(H_\theta|_{[0,N-1]}-(E+i\eta))^{-1}$ is the Green's function.

For $\eta>0$ small, let $m\in \N$ be chosen such that 
\begin{align}\label{eq:choose_eta}
\eta\sim_{\kappa,\beta,v,\omega,\delta}\, \exp(-m(4\kappa(E)\beta+12C_{\omega,v}\delta+2\delta^4))\leq \exp(-4m(\kappa(E)\beta(\omega)+2C_{\omega,v}\delta)).
\end{align}
We say that $k$ is {\it good}, denoted by $k\in \mathcal{Q}_m$, if 
\begin{align}\label{def:k_not_bad}
    \theta\notin (\mathcal{B}_{2m}^{\eta}+m\omega-k\omega).
\end{align}
Otherwise, we say $k$ is {\it bad}.
Since $\eta$ satisfies \eqref{eq:eta_small} for $2m$ in place of $m$, see \eqref{eq:choose_eta} above, we have by Corollary~\ref{cor:LDT} that
\begin{align}\label{eq:B2m_eta_small}
    \mathrm{mes}(\mathcal{B}_{2m}^{\eta})\leq e^{-2\delta^4 m}.
\end{align}

Note that \eqref{eq:M=PPPP} together with \eqref{def:k_not_bad} implies for each good $k\in \mathcal{Q}_m$, there exists 
$m_k\in \{2m,2m-1,2m-2\}$ and $s_k\in \{0,1\}$ such that
\begin{align}\label{eq:P_mk_good}
\frac{1}{2m}\log |P^{E+i\eta}_{m_k}(\theta+k\omega+s_k\omega-m\omega)|
\geq &u_{2m,E+i\eta}(\theta+k\omega-m\omega)-\frac{\log \sqrt{2}}{2m}\\
\geq &L(E)-2\kappa(E)\beta-3C_{\omega,v}\delta.
\end{align}
For each good $k\in \mathcal{Q}_m$, we let $[k+s_k-m,k+s_k-m+m_k-1]=:[a_k,b_k]$.

Distinguish the good and bad $k$'s in \eqref{eq:dN<G}, and bounding the bad terms using the trivial bound $|G^{E+i\eta}_{[0,N-1]}(\theta;k,k)|\leq \|G^{E+i\eta}_{[0,N-1]}(\theta)\|\leq \eta^{-1}$, we obtain
\begin{align}\label{eq:dN_good+bad}
    d_N(\theta,E-\eta,E+\eta)
    \leq &\frac{2\eta}{N}\left(\sum_{k=0}^{m-1} +\sum_{k=N-m}^{N-1}+\sum_{\substack{k=m\\ k\notin \mathcal{Q}_m }}^{N-1-m}+\sum_{\substack{k=m\\ k\in \mathcal{Q}_m}}^{N-1-m}\right) \mathrm{Im}(G^{E+i\eta}_{[0,N-1]}(\theta;k,k))\notag\\
    \leq &\frac{4m}{N}+\frac{2}{N}\sum_{k=m}^{N-1-m}\chi_{\mathcal{B}_{2m}^{\eta}+m\omega}(\theta+k\omega)+\frac{2\eta}{N}\sum_{\substack{k=m\\ k\in \mathcal{Q}_m}}^{N-1-m} \mathrm{Im}(G^{E+i\eta}_{[0,N-1]}(\theta;k,k)).
\end{align}

For $N$ large enough, we can bound by ergodic theorem and \eqref{eq:B2m_eta_small} that,
\begin{align}\label{eq:dN_bad}
    \frac{4m}{N}+\frac{2}{N}\sum_{k=m}^{N-1-m}\chi_{\mathcal{B}_{2m}^{\eta}+m\omega}(\theta+k\omega)\leq 2\cdot \mathrm{mes}(\mathcal{B}_{2m}^{\eta}+m\omega)+\eta\leq 2e^{-2\delta^4 m}+\eta.
\end{align}
Hence \eqref{eq:dN_good+bad} turns into
\begin{align}\label{eq:dN_good}
    d_N(\theta,E-\eta,E+\eta)\leq 2e^{-2\delta^4m}+\eta+2\eta\cdot \sup_{k\in \mathcal{Q}_m} \mathrm{Im}(G^{E+i\eta}_{[0,N-1]}(\theta;k,k)).
\end{align}

Next, we study $G^{E+i\eta}_{[0,N-1]}(\theta;k,k)$ for good $k$'s.
Let $\tilde{H}_{\theta}|_{[0,N-1],k}=H_\theta|_{[0,N-1]}-\Gamma_k$, where $\Gamma_k$ is a $N\times N$ matrix such that
\begin{align}
    \Gamma_k(j,\ell)=
    \begin{cases}
        1, \text{ if } (j,\ell)=(a_k,a_k-1), (a_k-1,a_k), (b_k,b_k+1), (b_k+1,b_k),\\
        0, \text{otherwise}
    \end{cases}
\end{align}
Clearly $\tilde{H}_{\theta}|_{[0,N-1],k}$ decouples into 
$$H_{\theta}|_{[a_k,b_k]} \text{ and } H_{\theta}|_{[0,N-1]\setminus [a_k,b_k]}.$$ 
Let $\tilde{G}_{[0,N-1],k}^{E+i\eta}(\theta)=(\tilde{H}_{\theta}|_{[0,N-1],k}-(E+i\eta))^{-1}$.

By the resolvent identity, we have
\begin{align}
    G_{[0,N-1]}^{E+i\eta}(\theta)-\tilde{G}_{[0,N-1],k}^{E+i\eta}(\theta)=-G_{[0,N-1]}^{E+i\eta}(\theta)\cdot \Gamma_k \cdot \tilde{G}_{[0,N-1],k}^{E+i\eta}(\theta).
\end{align}
This implies
\begin{align}\label{eq:G=tG}
    G_{[0,N-1]}^{E+i\eta}(\theta;k,k)
    =&\tilde{G}_{[0,N-1],k}^{E+i\eta}(\theta;k,k)-\sum_{j,\ell}G^{E+i\eta}_{[0,N-1]}(\theta;k,j)\cdot \Gamma_k(j,\ell) \cdot \tilde{G}^{E+i\eta}_{[0,N-1],k}(\theta;\ell,k)\notag\\
    =&\tilde{G}_{[a_k,b_k]}^{E+i\eta}(\theta;k,k)-G^{E+i\eta}_{[0,N-1]}(\theta;k,a_k-1)\cdot \tilde{G}^{E+i\eta}_{[a_k,b_k]}(\theta;a_k,k)\\
    &-G^{E+i\eta}_{[0,N-1]}(\theta;k,b_k+1)\cdot \tilde{G}^{E+i\eta}_{[a_k,b_k]}(\theta;b_k,k)\notag\\
     =&\tilde{G}_{[a_k,b_k]-k}^{E+i\eta}(\theta+k\omega;0,0)-G^{E+i\eta}_{[0,N-1]}(\theta;k,a_k-1)\cdot \tilde{G}^{E+i\eta}_{[a_k,b_k]-k}(\theta+k\omega;a_k-k,0)\\
    &-G^{E+i\eta}_{[0,N-1]}(\theta;k,b_k+1)\cdot \tilde{G}^{E+i\eta}_{[a_k,b_k]-k}(\theta+k\omega;b_k-k,0).
\end{align}
We bound the last two product terms above as follows:
\begin{align}\label{eq:G=tG_1}
    \max(|G^{E+i\eta}_{[0,N-1]}(\theta;k,a_k-1)|,\, |G^{E+i\eta}_{[0,N-1]}(\theta;k,b_k+1)|)\leq \|G^{E+i\eta}_{[0,N-1]}(\theta)\|\leq \eta^{-1},
\end{align}
and using \eqref{eq:P_mk_good} and Cramer's rule, we have for each good $k$ that (recall $a_k=k+s_k-m$),
\begin{align}\label{eq:G=tG_2}
    |\tilde{G}^{E+i\eta}_{[a_k,b_k]-k}(\theta+k\omega;a_k-k,0)|
    =&\frac{|P^{E+i\eta}_{b_k-k}(\theta+(k+1)\omega)|}{|P_{m_k}^{E+i\eta}(\theta+a_k\omega)|}\\
    \leq &e^{m(L(E+i\eta)+\delta)}\cdot e^{-2m(L(E)-2\kappa(E)\beta-3C_{\omega,v}\delta)}\\
    \leq &e^{-m(L(E)-4\kappa(E)\beta-8C_{\omega,v}\delta)},
\end{align}
and a similar estimate holds for $|\tilde{G}^{E+i\eta}_{[a_k,b_k]-k}(\theta+k\omega;b_k-k,0)|$ as well.
Combining \eqref{eq:G=tG_1}, \eqref{eq:G=tG_2} with \eqref{eq:G=tG}, we have
\begin{align}\label{eq:G=tG_n}
    |G^{E+i\eta}_{[0,N-1]}(\theta;k,k)|\leq |\tilde{G}^{E+i\eta}_{[a_k,b_k]-k}(\theta+k\omega;0,0)|+2\eta^{-1}e^{-m(L(E)-4\kappa(E)\beta-8C_{\omega,v}\delta)}.
\end{align}

Once again, by Cramer's rule, 
\begin{align}\label{eq:tG_est}
    |\tilde{G}^{E+i\eta}_{[a_k,b_k]-k}(\theta+k\omega;0,0)|
    =&\frac{|P_{k-a_k}^{E+i\eta}(\theta+a_k\omega)|\cdot |P_{b_k-k}^{E+i\eta}(\theta+(k+1)\omega)|}{|P_{m_k}^{E+i\eta}(\theta+a_k\omega)|} \notag\\
    \leq &e^{2m(L(E+i\eta)+\delta)}\cdot e^{-2m(L(E)-2\kappa(E)\beta-3C_{\omega,v}\delta)}\notag\\
    \leq &e^{4m(\kappa(E)\beta+3C_{\omega,v}\delta)},
\end{align}
where we used \eqref{eq:P_mk_good} again to estimate the denominator.
Combining \eqref{eq:G=tG_n} with \eqref{eq:tG_est}, we have
\begin{align}
    \sup_{k\in \mathcal{Q}_m}\mathrm{Im}(G^{E+i\eta}_{[0,N-1]}(\theta;k,k))\leq e^{4m(\kappa(E)\beta+3C_{\omega,v}\delta)}+2\eta^{-1}e^{-m(L(E)-4\kappa(E)\beta-8C_{\omega,v}\delta)}.
\end{align}
Plugging this in \eqref{eq:dN_good}, note we use the smallness of $\delta$ as in \eqref{def:delta_small}, we arrive at
\begin{align}
    d_N(\theta,E-\eta,E+\eta)
    &\leq 2e^{-2\delta^4m}+\eta+2\eta e^{4m(\kappa(E)\beta+3C_{\omega,v}\delta)}+4e^{-m(L(E)-4\kappa(E)\beta-8C_{\omega,v}\delta)}\\
    &\leq 3e^{-2\delta^4 m}+3\eta e^{4m(\kappa(E)\beta+3C_{\omega,v}\delta)}\\
    &\lesssim_{\kappa,\beta,v,\omega,\delta}\, \eta^{\frac{\delta^4}{2\kappa(E)\beta+7C_{\omega,v}\delta}},
\end{align} 
provided that $\eta\sim_{\kappa,\beta,v,\omega,\delta} \exp(-m(4\kappa(E)\beta+12C_{\omega,v}\delta+2\delta^4))$, as $N\to\infty$. 
This implies
\begin{align}
    \mathcal{N}(E-\eta, E+\eta)\lesssim_{\kappa,\beta,v,\omega,\delta}\eta^{\frac{\delta^4}{2\kappa(E)\beta+7C_{\omega,v}\delta}},
\end{align}
due to weak convergence, thus completes the proof of Theorem \ref{thm:main}. \qed

\section{Fourier coefficients of $u_{m,E}$}\label{sec:Riesz_vn}
Clearly \eqref{eq:L_eps=L_0+kappa} implies for any $0\leq \varepsilon_1<\varepsilon_2\leq \varepsilon_0$, 
\begin{align}\label{eq:Le2-Le1<=kappa}
    L(\omega,E,\varepsilon_2)-L(\omega,E,\varepsilon_1)
    =2\pi \kappa(\omega,E)(\varepsilon_2-\varepsilon_1).
\end{align}
Hence for any $\delta>0$, for $m$ large enough, we have
\begin{align}\label{eq:Lme2-Lme1<=kappa}
    L_m(\omega,E,\varepsilon_0)-L_m(\omega,E,\varepsilon_0/2)\leq \pi\varepsilon_0 \kappa(\omega,E)+\delta. 
\end{align}

One key new ingredient in this paper, in addition to the recent developments in \cite{HS1,Ha}, is the following effective bound on $C_{v,2}$, see \eqref{eq:Fourier_1/k}.
\begin{lemma}\label{lem:hat_vm_k}
    For any $\delta>0$, for $m$ large, we have
    \begin{align}
        |\hat{u}_{m,E}(k)|\leq \frac{\kappa(\omega,E)+\delta}{|k|}+C_{\omega,v,3} R^{-|k|/2}, \text{ for any } k\neq 0,
    \end{align}
    in which $C_{\omega,v,3}=C(\omega,v,\varepsilon_0)>0$ and $R=e^{2\pi \varepsilon_0}$.
\end{lemma}
\begin{proof}
As a preparation, we start off by discussing the Riesz representation of $u_{m,E}$.
\subsection{Control of Riesz mass of $u_{m,E}$}
The first subsection focuses on the control of Riesz mass of $u_{m,E}$. 
This part is analogous to \cite[Sec 5]{HS1}, with an improvement due to Lemma \ref{lem:Ih_constant} below.
\begin{lemma}\label{lem:int_Green}\cite[Lemma 3.2]{HS1}
For $1/R\leq r\leq R$ and $w$ in the annulus $\mathcal{A}_R$, we have
\begin{align}\label{eq:int_Green}
2\pi \int_0^{1}  G_R (re^{ 2\pi i\theta}, w)\, \mathrm{d}\theta 
=(2\log R)^{-1}
\begin{cases}
 \log (rR) \log|w/R|, \text{ if } |w|\geq r\\
\\
\log( r/ R) \log |w R|, \text{ if } |w|<r.
\end{cases}
\end{align}
\end{lemma}

\begin{theorem}\label{thm:Riesz_vn}
The Riesz representation of $u_{m,E}$ yields:
\begin{align}\label{eq:Riesz_vn}
u_{m,E}(z)=\int_{\mathcal{A}_R} 2\pi G_R(z,w)\, \mu_{m,E}(\mathrm{d}w)+h_{R,m,E}(z),
\end{align}
$G_R$ is as in \eqref{def:G} and the harmonic part satisfies $h_{R,m,E}=u_{m,E}$ on $\partial \mathcal{A}_R$.
We have for any $\delta>0$, for $m$ large enough,
\begin{align}\label{eq:Riesz_mass}
\mu_{m,E}(\mathcal{A}_{\sqrt{R}})\leq 2\kappa(\omega,E)+\delta.
\end{align}
\end{theorem}
\begin{proof}
Within the proof we shall omit $\omega$ in various notations.
We first show an analogue of \cite[Lemma 4.7]{Ha}.
\begin{lemma}\label{lem:Ih_constant}
    There exists $b\in \R$ such that
    \begin{align}
        \int_{\T}h_{R,m,E}(e^{2\pi i(\theta+i\varepsilon)})\, \mathrm{d}\theta\equiv b.
    \end{align}
\end{lemma}
\begin{proof}
Let
   \begin{align}\label{def:Iqn_GR_2'}
   I^G_m(E,\varepsilon):=&\int_{\T} \int_{\mathcal{A}_R} 2\pi G_{R}(e^{2\pi i(\theta+i\varepsilon)},w)\, \mu_{m,E}(\mathrm{d}w)\, \mathrm{d}\theta, \text{ and }\\
   I^h_m(E,\varepsilon):=&\int_{\T} h_{R,m,E}(e^{2\pi i(\theta+i\varepsilon)})\, \mathrm{d}\theta
\end{align}
Since $h_{R,m,E}$ is a harmonic function, its integral along an arbitrary circle centered at $0$ is a radial harmonic function, which implies
\begin{align}\label{eq:Ih=a_eps+b}
    I^h_m(E,\varepsilon)=a\varepsilon+b \text{ for some } a,b\in \R. 
\end{align}
Next, we show $I^h_m(E,\varepsilon)$ is an even function in $\varepsilon$. 
This follows from the evenness of $L_m(E,\varepsilon)$, see \eqref{eq:Lm_even}, and $I^G_m(E,\varepsilon)$, see \eqref{eq:IG_even} below.

Since the potential function $v(e^{2\pi i\theta})$ is real-valued for $\theta\in \T$, we have
\begin{align}
    v(z)=\overline{v(1/\overline{z})}, \text{ for } z\in \mathcal{A}_R.
\end{align}
This implies, since $E\in \R$,
\begin{align}\label{eq:u_{m,E}E_reflection}
    u_{m,E}(e^{2\pi i(\theta+i\varepsilon)})=u_{m,E}(e^{2\pi i(\theta-i\varepsilon)}), \text{ for } |\varepsilon|\leq \varepsilon_0,
\end{align}
and hence 
\begin{align}\label{eq:Lm_even}
L_m(E,\varepsilon)=L_m(E,-\varepsilon),
\end{align}
as well as
\begin{align}\label{eq:vmE_k_even}
    \int_{\T}u_{m,E}(e^{2\pi i(\theta+i\varepsilon)})e^{-2\pi ik\theta}\, \mathrm{d}\theta=\int_{\T}u_{m,E}(e^{2\pi i(\theta-i\varepsilon)})e^{-2\pi ik\theta}\, \mathrm{d}\theta.
\end{align}
We will use \eqref{eq:vmE_k_even} and \eqref{eq:G_k_even} below in Sec. \ref{sec:proof_vk}.

Since $u_{m,E}(z)=u_{m,E}(1/\overline{z})$, and $\Delta u_{m,E}=\mu_{m,E}$, the measure $\mu_{m,E}$ exhibits the same reflection symmetry as $u_{m,E}$ in \eqref{eq:u_{m,E}E_reflection}:
\begin{align}\label{eq:mu_{m,E}E_reflection}
\mu_{m,E}(\mathrm{d}z)=\mu_{m,E}(\mathrm{d}(1/\overline{z})).
\end{align}
This implies
\begin{align}\label{eq:GR_reflection}
\int_{\mathcal{A}_R}2\pi G_R(e^{2\pi i(\theta+i\varepsilon)},w)\, \mu_{m,E}(\mathrm{d}w)=&
\int_{\mathcal{A}_R}2\pi G_R(e^{2\pi i(\theta+i\varepsilon)},w)\, \mu_{m,E}(\mathrm{d}(1/\overline{w}))\\
=&\int_{\mathcal{A}_R}2\pi G_R(e^{2\pi i(\theta+i\varepsilon)},1/\overline{w})\, \mu_{m,E}(\mathrm{d}w)\\
=&\int_{\mathcal{A}_R}2\pi G_R(e^{2\pi i(\theta-i\varepsilon)},w)\, \mu_{m,E}(\mathrm{d} w),
\end{align}
in which we used $G_R(1/\overline{z}, 1/\overline{w})=G_R(z,w)$ as in \eqref{eq:GR_even}, in the last inequality.
Clearly \eqref{eq:GR_reflection} implies
\begin{align}\label{eq:IG_even}
    I^G_m(E,\varepsilon)=I^G_m(E,-\varepsilon),
\end{align}
as well as 
\begin{align}\label{eq:G_k_even}
    &\int_{\T} \left(\int_{\mathcal{A}_R}2\pi G_R(e^{2\pi i(\theta+i\varepsilon)},w)\, \mu_{m,E}(\mathrm{d}w)\right)e^{-2\pi ik\theta}\, \mathrm{d}\theta\\
    &\qquad =\int_{\T} \left(\int_{\mathcal{A}_R}2\pi G_R(e^{2\pi i(\theta-i\varepsilon)},w)\, \mu_{m,E}(\mathrm{d}w)\right)e^{-2\pi ik\theta}\, \mathrm{d}\theta
\end{align}
Combining \eqref{eq:Lm_even} with \eqref{eq:IG_even}, we have $I^h_m(E,\varepsilon)=I^h_m(E,-\varepsilon)$ which further implies 
\begin{align}
    I^h_m(E,\varepsilon)=b,
\end{align}
when combined with \eqref{eq:Ih=a_eps+b}.
This is the claimed result.
\end{proof}

Integrating \eqref{eq:Riesz_vn} along $w\in \mathcal{C}_{r_1}$ and $\mathcal{C}_{r_2}$, with $1< r_1<r_2\leq R$, and subtracting one from the other yields that 
\begin{align}\label{eq:RM_0}
&L_{m}(E,\frac{\log r_2}{2\pi})-L_{m}(E,\frac{\log r_1}{2\pi})\\
=&\int_0^{1} u_{m,E}(r_2 e^{2\pi i\theta})\, \mathrm{d}\theta-\int_0^{1} u_{m,E}(r_1 e^{2\pi i\theta})\, \mathrm{d}\theta\\
=&\int_{\mathcal{A}_R} \left(\int_0^{1}2\pi G_R(r_2 e^{2\pi i\theta},w)\, \mathrm{d}\theta-\int_0^{1} 2\pi G_R(r_1e^{2\pi i\theta},w)\, \mathrm{d}\theta\right)\, \mu_{m,E}(\mathrm{d}w)\\
&+\int_0^1 h_{R,m,E}(r_2 e^{2\pi i\theta})\, \mathrm{d}\theta-\int_0^1 h_{R,m,E}(r_1e^{2\pi i\theta})\, \mathrm{d}\theta\\
=&\int_{\mathcal{A}_R} \left(\int_0^{1}2\pi G_R(r_2 e^{2\pi i\theta},w)\, \mathrm{d}\theta-\int_0^{1} 2\pi G_R(r_1e^{2\pi i\theta},w)\, \mathrm{d}\theta\right)\, \mu_{m,E}(\mathrm{d}w),
\end{align}
where we used Lemma \ref{lem:Ih_constant} in the last equality. 

The rest of the proof is complete analogous to that of \cite[Theorem 5.1]{HS1}.
By \eqref{eq:Lme2-Lme1<=kappa}, for $m$ large, for $r_1=\sqrt{R}=e^{\pi\varepsilon_0}$ and $r_2=R=e^{2\pi\varepsilon_0}$, 
\begin{align}\label{eq:RM_0'}
L_{m}(E,\varepsilon_0)-L_{m}(E,\varepsilon_0/2)\leq \pi\varepsilon_0 \kappa(\omega,E)+\varepsilon_0\delta.
\end{align}
By Lemma \ref{lem:int_Green}, we have for any $1<r_1<r_2\leq R$,
\begin{align}\label{eq:RM_1}
&\int_0^{1}2\pi G_R(r_2 e^{2\pi i\theta},w)\, \mathrm{d}\theta-\int_0^{1}2\pi G_R(r_1e^{2\pi i\theta},w)\, \mathrm{d}\theta\\
=&\begin{cases}
\frac{\log r_2-\log r_1}{2\log R} \log \frac{|w|}{R}\, \text{ if } |w|\geq r_2\\
\\
\frac{\log (r_2/R)}{2\log R}\log |wR|-\frac{\log (r_1R)}{2\log R} \log \frac{|w|}{R},\, \text{ if } r_1\leq |w|\leq r_2\\
\\
\frac{\log r_2-\log r_1}{2\log R} \log (|w|R), \text{ if } |z|<r_1
\end{cases}
\end{align}
By \eqref{eq:mu_{m,E}E_reflection} and \eqref{eq:RM_1} we conclude that
\begin{align}\label{eq:RM_2}
\qquad\, \int_{|w|\geq r_2}\frac{\log r_2-\log r_1}{2\log R} \log \frac{|w|}{R}\, \mu_{m,E}(\mathrm{d}w)+\int_{|w|\leq r_2^{-1}}\frac{\log r_2-\log r_1}{2\log R}\log (|w|R)\, \mu_{m,E}(\mathrm{d}w)=0,
\end{align}
as well as
\begin{align}\label{eq:RM_3}
\qquad\qquad &\int_{r_1\leq |w|<r_2} \left(\frac{\log (r_2/R)}{2\log R}\log |wR|-\frac{\log (r_1R)}{2\log R} \log \frac{|w|}{R}\right)\, \mu_{m,E}(\mathrm{d}w) \notag\\
&+\int_{r_2^{-1}<|w|\leq r_1^{-1}} \frac{\log r_2-\log r_1}{2\log R} \log (|w|R)\, \mu_{m,E}(\mathrm{d}w) \notag\\
=&\int_{r_1\leq |w|<r_2} \left(\frac{\log (r_2/R)}{2\log R}\log |wR|-\frac{\log (r_1R)}{2\log R} \log \frac{|w|}{R}+\frac{\log r_2-\log r_1}{2\log R}\log \frac{R}{|w|}\right)\, \mu_{m,E}(\mathrm{d}w) \notag\\
=&\int_{r_1\leq |w|<r_2} \log \frac{r_2}{|w|}\, \mu_{m,E}(\mathrm{d}w),
\end{align}
and
\begin{align}\label{eq:RM_4}
\int_{1< |w|<r_1} \frac{\log r_2-\log r_1}{2\log R} \log (|w|R)\, \mu_{m,E}(\mathrm{d}w)+\int_{r_1^{-1}<|w|< 1}  \frac{\log r_2-\log r_1}{2\log R} \log (|w|R)\, \mu_{m,E}(\mathrm{d}w) \\
+\int_{|w|=1} \frac{\log r_2-\log r_1}{2}\, \mu_{m,E}(\mathrm{d}w)
=\frac{\log r_2-\log r_1}{2}\cdot \mu_{m,E}(\mathcal{A}_{r_1}).
\end{align}
Combining \eqref{eq:RM_0'}, \eqref{eq:RM_2}, \eqref{eq:RM_3} and~\eqref{eq:RM_4} with~\eqref{eq:RM_0}, one obtains
\begin{align}\label{eq:RM_6}
\pi\varepsilon_0\kappa(\omega,E)+\varepsilon_0\delta
\geq &\frac{\pi\varepsilon_0}{2} \mu_{m,E}(\mathcal{A}_{\sqrt{R}})+\int_{\sqrt{R}\leq |w|<R} \log \frac{R}{|w|}\, \mu_{m,E}(\mathrm{d}w)\notag\\
\geq &\frac{\pi\varepsilon_0}{2} \mu_{m,E}(\mathcal{A}_{\sqrt{R}}).
\end{align}
This yields
\begin{align}
    \mu_{m,E}(\mathcal{A}_{\sqrt{R}})\leq 2\kappa(\omega,E)+\delta,
\end{align}
as claimed.
\end{proof}

\subsection{Proof of Lemma \ref{lem:hat_vm_k}}\label{sec:proof_vk}

\ 

In this subsection we use the Riesz representation of $u_{m,E}$ on $\mathcal{A}_{\sqrt{R}}$, whose Riesz mass is controlled by Theorem \ref{thm:Riesz_vn}.
Throughout the rest of the paper, for simplicity, denote $\sqrt{R}=:R_1$.

We have
\begin{align}
    u_{m,E}(z)=G_{R_1,m,E}(z)+h_{R_1,m,E}(z),
\end{align}
where 
\begin{align}
    G_{R_1,m,E}(z)=\int_{\mathcal{A}_{R_1}}2\pi G_{R_1}(z,w)\, \mu_{m,E}(\mathrm{d}w).
\end{align}
We will prove the Fourier decay of $G_{R_1,m,E}$ and $h_{R_1,m,E}$ separately.
\begin{lemma}\label{lem:hat_GR1_k}
For $k\in \Z\setminus \{0\}$,
    \begin{align}
        |\hat{G}_{R_1,m,E}(k)|\leq \frac{\mu_{m,E}(\mathcal{A}_{R_1})}{2|k|}\left(1+R_1^{-|k|}\right),
    \end{align}
\end{lemma}
\begin{proof}
    It suffices to prove
    \begin{align}
        |(2\pi G_{R_1}(\cdot, w))^{\wedge}(k)|\leq \frac{1}{2|k|}\left(1+R_1^{-|k|}\right).
    \end{align}
    The following follows from explicit computations: if $w=re^{2\pi i\phi}$, then
    \begin{align}\label{eq:Fourier_k_log}
    \int_{\T} \log |e^{2\pi i\theta}-w| e^{-2\pi ik\theta}\, \mathrm{d}\theta
    =&\frac{1}{2\pi ik}\int_{\T} \frac{d}{d\theta}(\log |e^{2\pi i\theta}-w|) e^{-2\pi ik\theta}\, \mathrm{d}\theta \\
    =&\frac{1}{ik}\int_{\T} \frac{r\sin(2\pi(\theta-\phi))}{1-2r\cos(2\pi(\theta-\phi))+r^2} e^{-2\pi ik\theta}\, \mathrm{d}\theta\\
    =&\frac{1}{ik}e^{-2\pi ik\phi} \int_{\T} \frac{r\sin(2\pi\theta)}{1-2r\cos(2\pi \theta)+r^2} e^{-2\pi ik\theta}\, \mathrm{d}\theta\\
    =&\frac{\mathrm{sgn}(k)}{-2k}e^{-2\pi ik\phi} \begin{cases}
        r^{|k|}, \text{ if } r\leq 1,\\
        r^{-|k|}, \text{ if } r>1,
    \end{cases}
\end{align}
        Recall the definition of $G_{R_1}(z,w)$ in \eqref{def:G}, we have, note we convert each integral below into the form of \eqref{eq:Fourier_k_log}:
    \begin{align}
        (2\pi G_{R_1}(\cdot,w))^{\wedge}(k)=&\int_{\T} \log |e^{2\pi i\theta}-w| e^{-2\pi ik\theta}\, \mathrm{d}\theta \\
    &+\sum_{\ell=1}^{\infty} \int_{\T} \log |e^{2\pi i\theta}-\frac{1}{R_1^{4\ell}\overline{w}}| e^{-2\pi ik\theta}\, \mathrm{d}\theta\\
    &+\sum_{\ell=1}^{\infty} \int_{\T} \log |e^{2\pi i\theta}-\frac{w}{R_1^{4\ell}}| e^{-2\pi ik\theta}\, \mathrm{d}\theta\\
    &-\sum_{\ell=1}^{\infty} \int_{\T} \log |e^{2\pi i\theta}-\frac{w}{R_1^{4\ell-2}}|e^{-2\pi ik\theta}\, \mathrm{d}\theta\\
    &-\sum_{\ell=1}^{\infty} \int_{\T} \log |e^{2\pi i\theta}-\frac{1}{R_1^{4\ell-2}\overline{w}}|e^{-2\pi ik\theta}\, \mathrm{d}\theta\\
    \end{align}
    Directly applying \eqref{eq:Fourier_k_log}, with $w=re^{2\pi i\phi}$, yields, note that $R_1^{4\ell-2} \min(|w|,|w^{-1}|)>1$ for $w\in \mathcal{A}_{R_1}$,
    \begin{align}
    (2\pi G_{R_1}(\cdot,w))^{\wedge}(k)
    =&\frac{\mathrm{sgn}(k)}{-2k}e^{-2\pi ik\phi} 
    \begin{cases}
        r^{|k|}, \text{ if } r\leq 1\\
        r^{-|k|}, \text{ if } r>1
    \end{cases}\\
    &+\frac{\mathrm{sgn}(k)}{-2k}e^{-2\pi ik\phi}  \sum_{\ell=1}^{\infty} \left(\frac{1}{R_1^{4\ell |k|}}-\frac{1}{R_1^{(4\ell-2)|k|}}\right) (r^{|k|}+r^{-|k|})\\
    =&\frac{\mathrm{sgn}(k)}{-2k}e^{-2\pi ik\phi} 
    \begin{cases}
        r^{|k|}, \text{ if } r\leq 1\\
        r^{-|k|}, \text{ if } r>1
    \end{cases}\\
    &+\frac{\mathrm{sgn}(k)}{2k R_1^{|k|}}e^{-2\pi ik\phi} \frac{r^{|k|}+r^{-|k|}}{R_1^{|k|}+R_1^{-|k|}}
    \end{align}
    This implies 
    \begin{align}
    |(2\pi G_{R_1}(\cdot,w))^{\wedge}(k)|\leq \frac{1}{2|k|}(1+R_1^{-|k|}),
\end{align}
as claimed.
\end{proof}
As an immediate corollary of Theorem \ref{thm:Riesz_vn} and Lemma \ref{lem:hat_GR1_k}, we have:
\begin{corollary}\label{cor:hat_GR1_k}
For any $\delta>0$, for $m$ large, uniformly in $k\neq 0$, we have
    \begin{align}
        |2\pi\hat{G}_{R_1,m,E}(k)|\leq \frac{\kappa(\omega,E)+\delta}{|k|}(1+R_1^{-|k|}).
    \end{align}
\end{corollary}
\begin{lemma}\label{lem:hat_h_k}
For any $\delta>0$, for $m$ large enough, uniformly in $k\neq 0$, we have
    \begin{align}
        |\hat{h}_{R_1,m,E}(k)|\leq \frac{2L(\omega,E)+2\pi\varepsilon_0\kappa(\omega,E)+\delta}{R_1^{|k|}}.
    \end{align}
\end{lemma}
\begin{proof}
    In the proof we shall write $h_{R_1,m,E}$ as $h_E$ for simplicity.
We compute for $k\neq 0$ that, via integration by parts twice,
    \begin{align}
        I_k(E,r):=&\int_{0}^{2\pi}h_E(re^{i\theta})e^{-ik\theta}\, \mathrm{d}\theta\\
        =&\frac{1}{ik}\int_{\T} (\partial_{\theta} h_{E})(re^{i\theta}) e^{-ik\theta}\, \mathrm{d}\theta\\
        =&-\frac{1}{k^2}\int_{\T}(\partial_{\theta \theta} h_{E})(re^{i\theta}) e^{-ik\theta}\, \mathrm{d}\theta\\
        =&\frac{1}{k^2}\int_{\T} \left[r^2 (\partial_{rr}h_{E})(re^{i\theta})+r (\partial_{r}h_E)(re^{i\theta})\right] e^{-2\pi ik\theta}\, \mathrm{d}\theta\\
        =&\frac{1}{k^2} (r^2 I_k''(E,r)+r I_k'(E,r)).
    \end{align}
    The solution to this ODE is $I_k(E,r)=a_k(E) r^k+b_k(E) r^{-k}$, where $a_k(E),b_k(E)$ are constants. By \eqref{eq:vmE_k_even} and \eqref{eq:G_k_even}, we have $I_k(E,r)=I_k(E,1/r)$. Hence $a_k(E)=b_k(E)$, and 
    \begin{align}\label{eq:Ik=ak}
    I_k(E,r)=a_k(E)(r^k+r^{-k}).
    \end{align}
    Since the harmonic part is equal to $u_{m,E}$ on the boundary $\partial \mathcal{A}_{R_1}$, for $\delta>0$, we have
    \begin{align}\label{eq:hR1<L}
        0\leq h(R_1e^{2\pi i\theta})=u_{m,E}(R_1e^{2\pi i\theta})\leq L(\omega,E,-\varepsilon_0/2)+\delta,
    \end{align}
    for $m$ large enough, where we used the upper semi-continuity Lemma \ref{lem:upper_semi_cont} in the last inequality.
    Combining \eqref{eq:Ik=ak}, \eqref{eq:hR1<L} with \eqref{eq:Le2-Le1<=kappa}, we have
    \begin{align}
        |a_k(E)|(R_1^{|k|}+R_1^{-|k|})\leq 2\pi \sup_{\theta\in \T} |h(R_1e^{2\pi i\theta})|\leq 2\pi (L(\omega,E)+\pi\varepsilon_0\kappa(\omega,E)+\delta).
    \end{align}
    This implies
    \begin{align}
        |\hat{h}_{R_1,m,E}(k)|=(2\pi)^{-1} |I_k(E,1)|=\pi^{-1} |a_k(E)|\leq \frac{2L(\omega,E)+2\pi\varepsilon_0\kappa(\omega,E)+2\delta}{R_1^{|k|}},
    \end{align}
    as claimed.
\end{proof}
Combining Corollary \ref{cor:hat_GR1_k} with Lemma \ref{lem:hat_h_k}, we have proved Lemma \ref{lem:hat_vm_k}.
\end{proof}

\section{Large deviation estimate}
In this section we prove Theorem \ref{thm:LDT}.
Within the proof let $\beta=\beta(\omega)$, $\beta_n=\beta_n(\omega)$ for simplicity.
Let $n$ be such that 
\begin{align}\label{def:qn_m}
    \delta^{-2}(\beta+1)q_n\leq m<\delta^{-2}(\beta+1)q_{n+1}.
\end{align}
Let $Q=[\delta m]$ and 
\begin{align}
    u_{m,E}^{(Q)}(\theta):=\sum_{|j|<Q}\frac{Q-|j|}{Q^2}u_{m,E}(\theta+j\omega).
\end{align}
Note that $\hat{u}_{m,E}(0)=L_m(\omega,E)$.
Fourier expansion yields, note that the zeroth coefficient cancels:
\begin{align}
u_{m,E}(\theta)-L_m(\omega,E)=&u_{m,E}(\theta)-u^{(Q)}_{m,E}(\theta)=:U_1(\theta)\\
    &+\sum_{1\leq |k|\leq \delta^{-2}}\hat{u}_{m,E}(k)F_Q(k)e^{2\pi i k\theta}=:U_2(\theta)\\
    &+\sum_{\delta^{-2}<|k|<q_n}\hat{u}_{m,E}(k)F_Q(k)e^{2\pi i k\theta}=:U_3(\theta)\\
    &+\sum_{|\ell|=1}^{q_{n+1}/(4q_n)}\hat{u}_{m,E}(\ell q_n)F_Q(\ell q_n)e^{2\pi i \ell q_n\theta}=:U_4(\theta)\\
    &+\sum_{\ell=1}^{q_{n+1}/(4q_n)}\sum_{\ell q_n<|k|<(\ell+1)q_n}\hat{u}_{m,E}(k)F_Q(k)e^{2\pi i k\theta}=:U_5(\theta)\\
    &+\sum_{q_{n+1}/4\leq |k|<e^{4\delta^4 m}}\hat{u}_{m,E}(k)F_Q(k)e^{2\pi i k\theta}=:U_6(\theta)\\
    &+\sum_{|k|\geq e^{4\delta^4 m}}\hat{u}_{m,E}(k)F_Q(k)e^{2\pi i k\theta}=:U_7(\theta).
\end{align}
Our key new estimates are for $U_4$. 

By \eqref{eq:vm_shift_invariance}, we have
\begin{align}\label{eq:U1_1}
\|U_1\|_{L^{\infty}(\T)}=\|u_{m,E}-u_{m,E}^{(Q)}\|_{L^{\infty}(\T)}\leq C_{v,1}\sum_{|j|<Q}\frac{(Q-|j|)|j|}{Q^2m}\leq  C_{v,1}\frac{Q}{m}\leq C_{v,1} \delta.
\end{align}
Regarding $U_2$, we have by \eqref{eq:Fourier_1/mk} that
\begin{align}\label{eq:U2_1}
    \|U_2\|_{L^{\infty}(\T)}\leq \sum_{1\leq |k|\leq \delta^{-2}} \frac{C_{v,1}}{4m\|k\omega\|}\leq \frac{C_{v,1}\delta^{-2}}{2m}\cdot \max_{1\leq |k|\leq \delta^{-2}} \frac{1}{\|k\omega\|}\leq \delta,
\end{align}
provided $m$ is large enough.
We have by \eqref{eq:FR_2} and \eqref{eq:Fourier_1/k} that
    \begin{align}\label{eq:U3_1}
        \|U_3\|_{L^{\infty}(\T)}\leq C_{v,2}\delta^2 \sum_{1\leq |k|<q_n/4}\frac{1}{1+Q^2\|k\omega\|_{\T}^2}\leq 2\pi C_{v,2} \delta^2 \frac{q_n}{Q}\leq 2\pi C_{v,2}\delta^3,
    \end{align}
    where we used $Q=\delta m\geq \delta^{-1}q_n$, due to \eqref{def:qn_m}.
    Regarding $U_5$, we have by \eqref{eq:FR_3} and \eqref{eq:Fourier_1/k} that
\begin{align}\label{eq:U5_1}
    \|U_5\|_{L^{\infty}(\T)}
    \leq &2C_{v,2}\sum_{\ell=1}^{q_{n+1}/(4q_n)} \sum_{k\in (\ell q_n, (\ell+1)q_n)}\frac{1}{\ell q_n} \frac{1}{1+Q^2\|k\omega\|^2}
\end{align}
We have, since $\ell\leq [q_{n+1}/(4q_n)]=a_{n+1}/4$, that
\begin{align}\label{eq:ellqn_small}
    \|\ell q_n\omega\|\leq \frac{a_{n+1}}{4} \|q_n\omega\|\leq \frac{1}{4}\|q_{n-1}\omega\|,
\end{align}
where we used \eqref{eq:qn-1_norm=qn+qn+1} in the last inequality.
This implies for any $k\in (\ell q_n, (\ell+1)q_n)$ that
\begin{align}\label{eq:U5_11}
    \|k\omega\|\geq \|(k-\ell q_n)\omega\|-\|\ell q_n\omega\|\geq \|q_{n-1}\omega\|-\frac{1}{4}\|q_{n-1}\omega\|\geq \frac{3}{4}\|q_{n-1}\omega\|,
\end{align}
in which we used $0<|k-\ell q_n|<q_n$ and hence $\|(k-\ell q_n)\omega\|\geq \|q_{n-1}\omega\|$.

Next, we divide $(\ell q_n,(\ell+1)q_n)$ into the following two subsets:
\begin{align}
\begin{cases}
    K_1:=\{k\in (\ell q_n, (\ell+1)q_n): k\omega-[k\omega]\in (0,1/2)\},\\
    K_2:=\{k\in (\ell q_n, (\ell+1)q_n): k\omega-[k\omega]\in (1/2,1)\}.
\end{cases}
\end{align}
For $\{k_1\neq k_2\}\subset K_1$, we have
\begin{align}\label{eq:U5_12}
    |\|k_1\omega\|-\|k_2\omega\||=\|(k_1-k_2)\omega\|\geq \|q_{n-1}\omega\|.
\end{align}
Similarly for $\{k_1\neq k_2\}\subset K_2$, we also have $|\|k_1\omega\|-\|k_2\omega\||\geq \|q_{n-1}\omega\|$.
Combining \eqref{eq:U5_11} with \eqref{eq:U5_12}, we have that $\{\|k\omega\|\}_{k\in K_1}$ (same for $K_2$) are at least $\|q_{n-1}\omega\|$ spaced with the smallest term being at least $3\|q_{n-1}\omega\|/4$.
This implies, noting $\|q_{n-1}\omega\|\geq 1/(2q_n)$,
\begin{align}\label{eq:U5_13}
    \sum_{s=1}^2 \sum_{k\in K_s}\frac{1}{1+Q^2\|k\omega\|^2}
    \leq &2\sum_{j=1}^{q_n}\frac{1}{1+Q^2 j^2 (\|q_{n-1}\omega\|/2)^2} \notag\\
    \leq &2\sum_{j=1}^{q_n}\frac{1}{1+Q^2 j^2/(4q_n)^2} \notag\\
    \leq &2\int_0^{\infty} \frac{1}{1+Q^2 x^2/(4q_n)^2}\, \mathrm{d}x \notag\\
    = &\frac{8q_n}{Q}\int_0^{\infty} \frac{1}{1+x^2}\, \mathrm{d}x \notag\\
    = &\frac{4\pi q_n}{Q}.
\end{align}
Combining \eqref{eq:U5_1} with \eqref{eq:U5_13}, we have
\begin{align}\label{eq:U5_2}
    \|U_5\|_{L^{\infty}(\T)}
    \leq \frac{8\pi C_{v,2}}{Q}\sum_{\ell=1}^{q_{n+1}}\frac{1}{\ell}
    \leq 8\pi C_{v,2}\frac{\log q_{n+1}}{Q}\leq 8\pi C_{v,2} \beta_n \frac{q_n}{Q}\leq 8\pi C_{v,2}\delta,
\end{align}
where we used $Q\geq \delta m\geq \delta^{-1}(1+\beta)q_n>\delta^{-1}\beta_n q_n$, for $n$ large.

We also have by \eqref{eq:FR_3} and \eqref{eq:Fourier_1/k} that
    \begin{align}\label{eq:U6_1}
        \|U_6\|_{L^{\infty}(\T)}
        \leq &8C_{v,2}\sum_{\ell=1}^{4e^{4\delta^4 m}/q_{n+1}} \frac{1}{\ell q_{n+1}}\sum_{k\in [\ell q_{n+1}/4, (\ell+1)q_{n+1}/4)} \frac{1}{1+Q^2\|k\omega\|^2}\notag\\
        \leq &8C_{v,2}\sum_{\ell=1}^{4e^{4\delta^4 m}/q_{n+1}}\frac{1}{\ell q_{n+1}} \left(2+2\pi \frac{q_{n+1}}{Q}\right)\notag\\
        \leq &8C_{v,2}\left(\frac{10\delta^4 m}{q_{n+1}}+\frac{10\pi \delta^4 m}{Q}\right)\notag\\
        \leq &80C_{v,2}((\beta+1)\delta^2+\pi\delta^3)\notag\\
        \leq &400C_{v,2}\delta,
    \end{align}
    where we used $m\leq \delta^{-2}(\beta+1)q_{n+1}$, $Q=[\delta m]$ and $\delta<1/(\beta+1)$.
    
Regarding $U_7$, we have $L^2$ estimate by \eqref{eq:Fourier_1/k} that
\begin{align}\label{eq:U7_1}
    \|U_7\|_{L^2(\T)}^2=\sum_{|k|>e^{4\delta^4 m}} |\hat{u}_{m,E}(k)|^2 |F_Q(k)|^2\leq C_{v,2}^2 \sum_{|k|>e^{4\delta^4 m}} \frac{1}{|k|^2}\leq 2C_{v,2}^2 e^{-4\delta^4 m}.
\end{align}
The estimate of $U_4$ is new and is the key to Theorem \ref{thm:LDT}.
This lemma is the only place that requires the improved explicit bound in Lemma \ref{lem:hat_vm_k}.
\begin{lemma}\label{lem:LDT_U4}
For any small $\delta>0$, for $m$ large enough, we have
    \begin{align}
        \|U_4\|_{L^{\infty}(\T)}\leq 2\kappa(\omega,E) \beta(\omega)+C_{\omega,v,4}\delta.    \end{align}
\end{lemma}
\begin{proof}
We write $\kappa(\omega,E)$ as $\kappa$ for simplicity.
We further decompose
\begin{align}
    U_4(\theta)
    =&\sum_{|\ell|=1}^{[q_{n+1}^{1-\delta}/Q]} \hat{u}_{m,E}(\ell q_n) F_Q(\ell q_n)e^{2\pi i \ell q_n}\,\, (=:U_{4,1}(\theta))\\
    &+\sum_{|\ell|=[q_{n+1}^{1-\delta}/Q]+1}^{q_{n+1}/(4q_n)} \hat{u}_{m,E}(\ell q_n) F_Q(\ell q_n)e^{2\pi i \ell q_n}\,\, (=:U_{4,2}(\theta))
\end{align}
Note that by \eqref{eq:ellqn_small}, $\|\ell q_n\omega\|=|\ell| \|q_n\omega\|\geq |\ell|/(2q_{n+1})$.
This implies, by Lemma \ref{lem:hat_vm_k},
\begin{align}\label{eq:U4_22}
    \|U_{4,2}\|_{L^{\infty}(\T)}
    \leq &2\sum_{\ell=[q_{n+1}^{1-\delta}/Q]+1}^{q_{n+1}/(4q_n)} \left(\frac{\kappa+\delta}{\ell q_n} \frac{2}{1+Q^2\|\ell q_n\omega\|^2}+C_{\omega,v,3}e^{-\pi \varepsilon_0\ell q_n}\right)\notag\\
    \leq &\frac{4(\kappa+\delta)}{q_n}\left(\sum_{\ell=[q_{n+1}^{1-\delta}/Q]+1}^{q_{n+1}/(4q_n)}\frac{1}{\ell(1+Q^2\ell^2/(2q_{n+1})^2)}\right)+3C_{\omega,v,3}e^{-\pi \varepsilon_0 q_n}\notag\\
    \leq &\frac{4(\kappa+\delta)}{q_n}\left(1+\int_{[q_{n+1}^{1-\delta}/Q]+1}^{\infty} \frac{1}{x(1+Q^2x^2/(2q_{n+1})^2)}\, \mathrm{d}x\right)+3C_{\omega,v,3}e^{-\pi \varepsilon_0 q_n}\notag\\
    \leq &\frac{4(\kappa+\delta)}{q_n}\left(1+\int_{q_{n+1}^{-\delta}/2}^{\infty} \frac{1}{x(1+x^2)}\, \mathrm{d}x\right)+3C_{\omega,v,3}e^{-\pi \varepsilon_0 q_n}\notag\\
    \leq &\frac{4(\kappa+\delta)}{q_n}\left(1+\int_{q_{n+1}^{-\delta}/2}^{1} \frac{1}{x}\, \mathrm{d}x+\int_1^{\infty} \frac{1}{1+x^2}\, \mathrm{d}x\right)+3C_{\omega,v,3}e^{-\pi \varepsilon_0 q_n}\notag\\
    \leq &4(\kappa+\delta)(\delta \beta_n+4q_n^{-1})+3C_{\omega,v,3}e^{-\pi \varepsilon_0 q_n}\notag\\
    \leq &4(\kappa+\delta)\delta \beta_n+\delta,
\end{align}
for $m$ large enough (when $m$ is large, $q_n$ and $n$ are also large).

For $1\leq \ell\leq q_{n+1}^{1-\delta}/Q$, we give a more precise estimate of $F_Q(\ell q_n)$ than \eqref{eq:FR_1}.
Note that for such $\ell$'s, and any $|j|<Q$, we have
\begin{align}
    \|\ell jq_n\omega\|=\ell |j|\|q_n\omega\|\leq \ell Q \|q_n\omega\|\leq q_{n+1}^{-\delta}.
\end{align}
Hence $e^{2\pi i\ell jq_n\omega}=1+O(q_{n+1}^{-\delta})$.
Going back to its definition of $F_Q$, we have
\begin{align}\label{eq:FR_new}
    F_Q(\ell q_n)
    =&\sum_{|j|<Q}\frac{Q-|j|}{Q^2} e^{2\pi i \ell jq_n\omega} \notag\\
    =&\sum_{|j|<Q}\frac{Q-|j|}{Q^2} (1+O(q_{n+1}^{-\delta})) \notag\\
    =&1+O(q_{n+1}^{-\delta}).
\end{align}
Combining Lemma \ref{lem:hat_vm_k} with \eqref{eq:FR_new}, we have
\begin{align}\label{eq:U4_11}
\|U_{4,1}\|_{L^{\infty}(\T)}
\leq &2\sum_{\ell=1}^{[q_{n+1}^{1-\delta}/Q]}\left( \frac{\kappa+\delta}{\ell q_n}(1+O(q_{n+1}^{-\delta}))+C_{\omega,v,4} e^{-\pi \ell q_n\varepsilon_0}\right)\notag\\
\leq &2(\kappa+\delta)\beta_n (1-\delta)(1+O(q_{n+1}^{-\delta}))+3C_{\omega,v,4}e^{-\pi q_n\varepsilon_0}\notag\\
\leq &2(\kappa+\delta)\beta_n (1-\delta/2)+\delta.
\end{align}
Therefore, combining \eqref{eq:U4_22} with \eqref{eq:U4_11}, we have, using $\delta<(1+\beta)^{-1}$,
\begin{align}
    \|U_4\|_{L^{\infty}(\T)}
    \leq &2(\kappa+\delta)\beta_n (1+2\delta)+2\delta\\
    \leq &2(\kappa+\delta)\beta(1+3\delta)+2\delta\\
    \leq &2\kappa\beta+2(\beta(1+3(1+\beta)^{-1})+1)\delta,
\end{align}
for $m$ large enough (which implies $n$ is large), as claimed.
\end{proof}

Combining \eqref{eq:U1_1}, \eqref{eq:U2_1}, \eqref{eq:U3_1}, \eqref{eq:U5_1}, \eqref{eq:U6_1} with Lemma \ref{lem:LDT_U4}, we have for $m$ large enough,
\begin{align}
    \|u_{m,E}-L_m-U_7\|_{L^{\infty}(\T)}
    \leq \sum_{j=1}^6 \|U_j\|_{L^{\infty}(\T)}
    \leq 2\kappa\beta+C_{\omega,v,5}\delta,
\end{align}
where $C_{\omega,v,5}$ depends on $C_{v,1}$, $C_{v,2}$ and $C_{\omega,v,4}$.
Combining this with \eqref{eq:U6_1} and the Chebyshev's inequality, we have
\begin{align}
    &\mathrm{mes}\left\{\theta: |u_{m,E}(\theta)-L_m(\omega,E)|>2\kappa\beta+2C_{\omega,v,5}\delta\right\}\\
    \leq & \mathrm{mes}\{\theta: |U_7(\theta)|>C_{\omega,v,5}\delta\}\leq \frac{1}{C_{\omega,v,5}^2\delta^2}\|U_7\|_{L^2(\T)}^2\leq \frac{2\delta^{-2}C_{v,2}}{C_{\omega,v,5}^2}e^{-4\delta^4 m}\leq e^{-\delta^4 m},
\end{align}
for $m$ large enough as claimed.
\qed

\end{document}